\documentclass[reqno,11pt,letterpaper]{amsart}

\usepackage{lipsum}
\usepackage{amsmath}
\usepackage{amssymb}
\usepackage{amsthm}
\usepackage{mathrsfs}
\usepackage{accents}
\usepackage{calc}
\usepackage{arydshln}
\usepackage{upgreek}
\usepackage{slashed}
\usepackage{xifthen}
\usepackage{graphicx}
\usepackage{longtable}
\usepackage[inline]{enumitem}

\usepackage{xcolor}
\definecolor{winered}{rgb}{0.6,0,0}
\definecolor{lessblue}{rgb}{0,0,0.7}

\usepackage[pdftex,colorlinks=true,linkcolor=winered,citecolor=lessblue,urlcolor=lessblue,breaklinks=true,bookmarksopen=true]{hyperref}

\makeatletter
\newcommand{\myitem}[3]{\item[#2]\def\@currentlabel{#3}\label{#1}}
\makeatother

\usepackage{titletoc}

\makeatletter

\def\@tocline#1#2#3#4#5#6#7{
\begingroup
  \par
    \parindent\z@ \leftskip#3 \relax \advance\leftskip\@tempdima\relax
                  \rightskip\@pnumwidth plus 4em \parfillskip-\@pnumwidth
    \ifcase #1 % sections
       \vskip 0.6em \hskip 0em % add a little vspace before
       \or
       \or \hskip 0em % subsections
       \or \hskip 1em % subsubsections
    \fi%
    #6
    \nobreak\relax{\leavevmode\leaders\hbox{\,.}\hfill}
    \hbox to\@pnumwidth {\@tocpagenum{#7}}
  \par
\endgroup
}

 \def\l@section{\@tocline{0}{0pt}{0pc}{}{}}

\renewcommand{\tocsection}[3]{%
  \indentlabel{\@ifnotempty{#2}{ % for numbered sections
    \ignorespaces\bfseries{#2. #3}}}
  \indentlabel{\@ifempty{#2}{\ignorespaces\bfseries{#3}}{}} % for unnumbered sections
    \vspace{1.5pt}}

\renewcommand{\tocsubsection}[3]{%
  \indentlabel{\@ifnotempty{#2}{
    \ignorespaces#2. #3}}
  \indentlabel{\@ifempty{#2}{\ignorespaces #3}{}}
    \vspace{1.5pt}}

\renewcommand{\tocsubsubsection}[3]{%
  \indentlabel{\@ifnotempty{#2}{
    \ignorespaces#2. #3}}
  \indentlabel{\@ifempty{#2}{\ignorespaces #3}{}}
    \vspace{1.5pt}}

\makeatother

\makeatletter
\def\@nomenstarted{0}
\newlength{\@nomenoldtabcolsep}

\newcommand{\nomenstart}
  {%
    \def\@nomenstarted{1}%
    \setlength{\@nomenoldtabcolsep}{\tabcolsep}%
    \setlength{\tabcolsep}{3.5pt}%
    \begin{longtable}{p{0.11\textwidth} p{0.86\textwidth}}%found by hand
  }

\newcommand{\nomenitem}[2]{%
    \ifcase\@nomenstarted%
      \or % if nomenstarted=1, do nothing
      \or \\ % if nomenstarted=2, add newline to previous one
    \fi%
    #1\,{\leavevmode\leaders\hbox{\,.}\hfill} & #2%
    \def\@nomenstarted{2}%
  }%
\newcommand{\nomenend}
  {\\%
      \end{longtable}%
      \setlength{\tabcolsep}{\@nomenoldtabcolsep}%
      \def\@nomenstarted{0}%
  }
\makeatother

\makeatletter
\newcommand{\BIG}{\bBigg@{3.5}}
\newcommand{\vast}{\bBigg@{4}}
\newcommand{\Vast}{\bBigg@{5}}
\newcommand{\VAST}[1]{\bBigg@{#1}}
\makeatother

\allowdisplaybreaks

\numberwithin{equation}{section}
\numberwithin{figure}{section}
\newtheorem{thm}{Theorem}[section]

\newtheorem{prop}[thm]{Proposition}

\newtheorem*{thm*}{Theorem}
\newtheorem*{prop*}{Proposition}
\newtheorem*{cor*}{Corollary}
\newtheorem*{conj*}{Conjecture}

\theoremstyle{definition}

\theoremstyle{remark}

\makeatletter
\newcommand{\fakephantomsection}{%
  \Hy@MakeCurrentHref{\@currenvir.\the\Hy@linkcounter}
  \Hy@raisedlink{\hyper@anchorstart{\@currentHref}\hyper@anchorend}%
  \Hy@GlobalStepCount\Hy@linkcounter%
}
\makeatother

\newcommand{\mc}{\mathcal}

\newcommand{\cC}{\mc C}

\newcommand{\cL}{\mc L}

\newcommand{\cO}{\mc O}

\newcommand{\cU}{\mc U}
\newcommand{\cV}{\mc V}

\newcommand{\ms}{\mathscr}

\newcommand{\sR}{\ms R}

\newcommand{\N}{\mathbb{N}}
\newcommand{\R}{\mathbb{R}}

\newcommand{\Sph}{\mathbb{S}}

\newcommand{\sfG}{\mathsf{G}}

\newcommand{\Err}{{\mathrm{Err}}{}}

\newcommand{\ran}{\operatorname{ran}}

\newcommand{\Id}{\operatorname{Id}}

\newcommand{\tr}{\operatorname{tr}}

\newcommand{\dv}{\operatorname{div}}

\newcommand{\Ups}{\Upsilon}

\newcommand{\hra}{\hookrightarrow}
\newcommand{\la}{\langle}

\newcommand{\pa}{\partial}
\newcommand{\dd}{{\mathrm d}}
\newcommand{\ra}{\rangle}

\newcommand{\ul}[1]{\underline{#1}{}}

\newcommand{\wt}{\widetilde}

\newcommand{\Diff}{\mathrm{Diff}}

\newcommand{\loc}{{\mathrm{loc}}}
\newcommand{\CI}{\cC^\infty}
\newcommand{\CIdot}{\dot\cC^\infty}

\newcommand{\Ric}{\mathrm{Ric}}

\newcommand{\openbigpmatrix}[1]
  {%
    \def\@bigpmatrixsize{#1}%
    \addtolength{\arraycolsep}{-#1}%
    \begin{pmatrix}%
  }
\newcommand{\closebigpmatrix}
  {%
    \end{pmatrix}%
    \addtolength{\arraycolsep}{\@bigpmatrixsize}%
  }

\newlength{\enummargin}

\DeclareGraphicsExtensions{.mps}

\makeatletter
\newcommand*{\fwbw}[1]{\expandafter\@fwbw\csname c@#1\endcsname}
\newcommand*{\@fwbw}[1]{\ifcase #1 \or {\rm fw}\or {\rm bw}\fi}
\AddEnumerateCounter{\fwbw}{\@fwbw}
\makeatother

\begin{document}

%%%%%%%%%%%%%%%%%%%%%%%%%%%%%%%%%%%%%%%%%%%%%%%%%%%%%%%%%%%%%%%%%%%%%%
% title page
\title[Asymptotically de~Sitter metrics from scattering data]{Asymptotically de~Sitter metrics from scattering data in all dimensions}

\date{\today}

\begin{abstract}
  In spacetime dimensions $n+1\geq 4$, we show the existence of solutions of the Einstein vacuum equations which describe asymptotically de~Sitter spacetimes with prescribed smooth data at the conformal boundary. This provides a short alternative proof of a special case of a result by Shlapentokh-Rothman and Rodnianski, and generalizes earlier results by Friedrich and Anderson to all dimensions.
\end{abstract}

% 83C05: Einstein field equations
% 35C20: asymptotic expansions

\author{Peter Hintz}
\address{Department of Mathematics, ETH Z\"urich, R\"amistrasse 101, 8092 Z\"urich, Switzerland}
\email{peter.hintz@math.ethz.ch}

\maketitle

%%%%%%%%%%%%%%%%%%%%%%%%%%%%%%%%%%%%%%%%%%%%%%%%%%%%%%%%%%%%%%%%%%%%%%
\section{Introduction}
\label{SI}

We study the existence of Lorentzian metrics $g$ on $M=[0,1)_\tau\times X$, where $X$ is an $n$-dimensional manifold, $n\geq 3$, satisfying the Einstein vacuum equations
\begin{equation}
\label{EqI}
  \Ric(g) - n g = 0,
\end{equation}
which are \emph{asymptotically de~Sitter} metrics. This means that $g=\tau^{-2}(-\dd\tau^2+h)$ modulo terms which are more regular at $\tau=0$; here $h$ is a Riemannian metric on $X$. The nomenclature arises from the special case that $(X,h)$ is $\R^n$ with the Euclidean metric $\dd x^2$: then $([0,\infty)_\tau\times\R^n,\tau^{-2}(-\dd\tau^2+\dd x^2))$ describes (a piece of) de~Sitter space. Specifically, we start by considering metrics $g_0$ which have asymptotic expansions
\begin{equation}
\label{EqIExp}
  g_0 \sim -\frac{\dd\tau^2}{\tau^2}
    +
      \begin{cases}
        \tau^{-2}\sum_{j\geq 0} h_j(\tau)(\tau^n\log\tau)^j, & n\ \text{even}, \\
        \tau^{-2}(h_0(\tau) + \tau^n h_1(\tau)), & n\ \text{odd},
      \end{cases}
\end{equation}
as $\tau\to 0$; here the $h_j(\tau)=h_j(\tau,x;\dd x)$ (with $h_0(0)=h$) are symmetric 2-tensors on $X$ which are even functions of $\tau$. It was shown by Fefferman--Graham \cite{FeffermanGrahamAmbient,FeffermanGrahamAmbientBook} that upon specifying two pieces of smooth data
\[
  h:=h_0(0),\quad
  k:=\pa_\tau^n h_0(0) \in \CI(X;S^2 T^*X),
\]
for which one must require $\tr_h k=0$ and $\dv_h k=D$ where $D=0$ when $n$ is odd whereas $D$ is a certain 1-form when $n$ is even, the formal version of equation~\eqref{EqI} at $\tau=0$ (i.e.\ requiring only that $\Ric(g_0)-n g_0$ vanish to infinite order at $\tau=0$) uniquely determines all remaining Taylor coefficients of the $h_j$. See \cite[Theorems~4.8 and 3.10]{FeffermanGrahamAmbientBook}; we give a brief self-contained derivation of this result in Appendix~\ref{SF}.\footnote{Conversely, by Borel's lemma, there exists a log-smooth tensor $g_0$ with the required generalized Taylor expansion~\eqref{EqIExp}, and then $\Ric(g_0)-n g_0$ is smooth and vanishes to infinite order at $\tau=0$.} More precisely, Fefferman--Graham studied the existence of (generalized) \emph{ambient metrics} associated with $(X,h)$: these are Lorentzian metrics $\wt g$ on $(-1,1)_\rho\times(0,\infty)_t\times X$ which are homogeneous of degree $2$ with respect to dilations in $t$, which at $\rho=0$ restrict to the degenerate metric $t^2 h$, and for which $\Ric(\wt g)$ vanishes to infinite order at $\rho=0$. When the ambient metric satisfies an additional condition, called \emph{straight} in \cite[Proposition~2.4]{FeffermanGrahamAmbientBook}, one can quotient out by the dilation action in a clean manner and obtain an asymptotically de~Sitter metric when $\rho$ has one sign (we shall refer to this region as the `exterior of the light cone'), and an asymptotically hyperbolic metric when $\rho$ has the opposite sign. In \cite[\S4]{FeffermanGrahamAmbientBook} the case of asymptotically hyperbolic metrics is discussed, but the asymptotically de~Sitter case is completely analogous.

A natural question is whether every \emph{formal} solution $g_0$ of~\eqref{EqI}, with smooth data $h$ and $k$, can be corrected, by a smooth symmetric 2-tensor $g'(\tau,x;\dd x)$ vanishing to infinite order at $\tau=0$, to a \emph{true} solution $g=g_0+g'$ of~\eqref{EqI}. Rodnianski and Shlapentokh-Rothman describe in \cite[\S{1.4.2}]{RodnianskiShlapentokhRothmanSelfSimilar} how their results on asymptotically self-similar spacetimes imply the existence of a true solution, in the smooth category, of a generalized ambient metric with Fefferman--Graham asymptotics in the exterior of the light cone. In the straight setting, their result thus already provides an affirmative answer to this question. In view of the complexities encountered in \cite{RodnianskiShlapentokhRothmanSelfSimilar} (which proves a significantly more general result), we give in this paper an alternative, short and elementary argument following the ideas of \cite[\S3.3]{HintzGluedS}.

\begin{thm}[Main result]
\label{ThmI}
  Suppose $g_0$ is a Lorentzian metric on $M=[0,1)_\tau\times X$ of the form~\eqref{EqIExp} for which $\Ric(g_0)-n g_0$ is smooth and vanishes to infinite order at $\tau=0$. Then on a sufficiently small open neighborhood $\cU$ of $\{0\}\times X\subset M$, there exists a unique smooth symmetric 2-tensor $g'(\tau)=g'(\tau,x;\dd x)$ which vanishes to infinite order at $\tau=0$ so that
  \begin{equation}
  \label{EqIEinstein}
    \Ric(g_0+g') - n(g_0+g') = 0\quad\text{in}\quad \cU.
  \end{equation}
\end{thm}

The existence of a smooth correction $g'=g'(\tau,x;\dd\tau,\dd x)=\cO(\tau^\infty)$ follows under significantly weaker assumptions on the structure of $g_0$, namely $\tau^2 g_0=-\dd\tau^2+h+\bar g'(\tau,x;\dd\tau,\dd x)$ where the coefficients of $\bar g'$ are bounded by some fixed power $\tau^\eta$, $\eta>0$, along with all their derivatives along $\tau\pa_\tau$ and $\tau\pa_x$. We construct the correction term $g'$ as the limit of backward solutions (i.e.\ imposing trivial data at $\tau=\delta$ and letting $\delta\searrow 0$) of a gauge-fixed version of the Einstein vacuum equations which are a system of quasilinear wave equations. (A similar procedure for linear and nonlinear wave equations was used in an analytically related setting by Petersen \cite{PetersenCauchyHorWave}.) We use a generalized harmonic gauge, namely a wave map/DeTurck gauge relative to background metric $g_0$. The resulting solution $g=g_0+g'$ satisfies the gauge condition to infinite order at $\tau=0$; but the gauge 1-form, measuring the failure of the gauge condition, satisfies a homogeneous wave equation on the asymptotically de~Sitter spacetime $(M,g)$. A unique continuation argument based on an energy estimate shows that this gauge 1-form must vanish, and thus $g$ solves the Einstein equations~\eqref{EqI}.

Our proof differs from previous arguments which treated the case of odd spatial dimensions $n\geq 3$. For $n=3$, Friedrich's work \cite{FriedrichDeSitterPastSimple} identifies $(h,k)$ with asymptotic data for the Einstein vacuum equations at the future conformal boundary of an (asymptotically) de~Sitter space, with $k$ having an interpretation as part of the rescaled Weyl tensor. In addition to proving the existence of solutions of~\eqref{EqI} attaining any given (smooth, or merely sufficiently regular) data $(h,k)$, Friedrich also proves, conversely, that small perturbations of non-characteristic initial data at $\tau=\tau_0>0$ evolve into asymptotically de~Sitter spacetimes, thus establishing a one-to-one correspondence between asymptotic data $(h,k)$ and asymptotically de~Sitter metrics solving~\eqref{EqI}. Anderson \cite{AndersonStabilityEvenDS} extended Friedrich's result to all odd $n\geq 3$. (See also \cite{RingstromEinsteinScalarStability} for forward stability results in general dimensions, which however do not yield a description of asymptotically de~Sitter spacetimes in terms of scattering data.)

We remark that in the case that $h,k$ are real-analytic, the convergence of the expansion~\eqref{EqIExp} was shown by Rendall \cite{RendallLambdaAsymptotics} (see \cite[Theorems~2 and 3]{RendallLambdaAsymptotics} for the formal power series construction and \cite[Theorem~6]{RendallLambdaAsymptotics} for the convergence statement). It also follows, by restriction to the exterior of the light cone in the straight setting, from the proof of convergence of the (generalized) Taylor expansion of the ambient metric by Kichenassamy \cite{KichenassamyFeffermanGraham}, which in turn improved upon the original result by Fefferman--Graham \cite{FeffermanGrahamAmbientBook}.

Finally, we remark that an elliptic problem related to the one considered here was solved in dimension $n\geq 3$ by Graham--Lee \cite{GrahamLeeConformalEinstein}; in this case only $h$ can be freely specified, and the existence of Poincar\'e--Einstein metrics close to the hyperbolic metric was shown when the datum $h$ is close to the standard metric on $\Sph^{n-1}$.

%%%%%%%%%%%%%%%%%%%%%%%%%%%%%%%%%%%%%%%%%%%%%%%%%%
\subsection*{Acknowledgments}

I would like to thank the organizers of the conference \textit{At the Interface of Asymptotics, Conformal Methods, and Analysis in General Relativity} on May 9--10, 2023, at the Royal Society in London for putting together an inspiring meeting, and Marc Mars and Piotr Chru\'sciel for stimulating discussions. I gratefully acknowledge the hospitality of the Erwin Schr\"odinger Institute in Vienna in June 2023 during the writing of this paper.

%%%%%%%%%%%%%%%%%%%%%%%%%%%%%%%%%%%%%%%%%%%%%%%%%%%%%%%%%%%%%%%%%%%%%%
\section{0-geometry}
\label{S0}

Let $M$ denote a smooth $(n+1)$-dimensional manifold with boundary $\pa M$; we assume that $\pa M$ is an embedded submanifold. By $\cV_0(M)$ we denote the space of 0-vector fields (or uniformly degenerate vector fields) \cite{MazzeoMelroseHyp}, consisting of all smooth vector fields vanishing at $\pa M$. In local coordinates $\tau\geq 0$, $x=(x^1,\ldots,x^n)\in\R^n$, elements of $\cV_0(M)$ are of the form $a(\tau,x)\tau\pa_\tau+\sum_{j=1}^n b^j(\tau,x)\tau\pa_{x^j}$ where $a,b^j$ are smooth. Locally finite sums of up to $m$-fold compositions of 0-vector fields yield the space $\Diff_0^m(M)$ of $m$-th order 0-differential operators. We introduce a vector bundle ${}^0 T M\to M$ with smooth frame $\tau\pa_\tau$, $\tau\pa_{x^j}$ ($j=1,\ldots,n$); over points $p\in M^\circ$ in the interior of $M$, the identity map induces an isomorphism ${}^0 T_p M\to T_p M$, but this map ceases to be injective for $p\in\pa M$. The dual bundle ${}^0 T^*M$ has a smooth frame $\frac{\dd\tau}{\tau}$, $\frac{\dd x^j}{\tau}$ ($j=1,\ldots,n$). An example of a Lorentzian signature section of the corresponding bundle $S^2\,{}^0 T^*M$ of symmetric 2-tensors is
\[
  \tau^{-2}\bigl(-\dd\tau^2 + h(x,\dd x)\bigr)
\]
where $h$ is a Riemannian metric on $\R^n$; compare this with~\eqref{EqIExp}.

We define ${}^0\cC^0(M)=L_\loc^\infty(M)\cap\cC^0(M^\circ)$. For $k\geq 1$, we inductively define ${}^0\cC^k(M)$ as the space of all $u\in{}^0\cC^{k-1}(M)$ so that $V u\in{}^0\cC^{k-1}(M)$ for all $V\in\cV_0(M)$. If $\tau\in\CI(M)$ is a boundary defining function, i.e.\ $\pa M=\tau^{-1}(0)$ and $\dd\tau\neq 0$ on $\pa M$, then for $\eta\in\R$ we also have weighted spaces $\tau^\eta\,{}^0\cC^k(M)=\{\tau^\eta u\colon u\in{}^0\cC^k(M)\}$. Elements of $\Diff_0^m(M)$ act continuously between such spaces; one can also consider operators of class $\tau^\eta\,{}^0\cC^k\Diff_0^m(M)$ which are sums of operators of the form $a P$ where $a\in\tau^\eta\,{}^0\cC^k(M)$ and $P\in\Diff_0^m(M)$. Note next that $\tau^\eta\,{}^0\cC^0(M)\subset\cC^0(M)$ for $\eta>0$. More generally, since a smooth vector field on $M$ is $\tau^{-1}$ times a 0-vector field,
\[
  \tau^{l+\eta}\,{}^0\cC^k(M) \subset \cC^k(M),\qquad k\leq l\in\N_0,\quad \eta>0.
\]
In particular, we have $\bigcap_{l,k}\tau^{l+\eta}\,{}^0\cC^k(M)=\CIdot(M)$, the space of smooth functions on $M$ vanishing to infinite order at $\pa M$.

We similarly have $L^2$-based Sobolev spaces, which naturally arise in energy estimates. We define $H_{0,\loc}^0(M):=L_\loc^2(M)$ to be the local $L^2$-space with respect to a positive 0-density, i.e.\ in local coordinates $\tau\geq 0$, $x\in\R^n$ as above a smooth positive multiple of $|\frac{\dd\tau}{\tau}\frac{\dd x}{\tau^n}|=|\frac{\dd\tau}{\tau}\frac{\dd x^1}{\tau}\cdots\frac{\dd x^n}{\tau}|$; and then $H_{0,\loc}^k(M)$ is defined to consist of all $u\in H_{0,\loc}^{k-1}(M)$ so that $V u\in H_{0,\loc}^{k-1}(M)$ for all $V\in\cV_0(M)$. If, in local coordinates, near any point $(\tau_0,x_0)$, $\tau_0>0$, one introduces the coordinates $T=\frac{\tau-\tau_0}{\tau_0}$ and $X=\frac{x-x_0}{\tau_0}$, then for bounded $|T|+|X|$, the vector fields $\tau\pa_\tau$, $\tau\pa_x$ on the one hand and $\pa_T$, $\pa_X$ on the other hand are linear combinations of each other with uniformly bounded and smooth (in $T,X$) coefficients; thus, Sobolev embedding on the unit ball in $\R^{n+1}$ implies $\tau^\alpha H_0^m(M)\subset\tau^\alpha\,{}^0\cC^k(M)$ for $\alpha\in\R$ and $m>\frac{n+1}{2}+k$.

When $M$ is compact, the spaces ${}^0\cC^k(M)$ are Banach spaces, with the norm of $u$ given by the maximum of the sup norms of $u$ and all its up to $k$-fold derivatives along the elements of a fixed finite spanning set of $\cV_0(M)$; similarly for the spaces $\tau^\eta\,{}^0\cC^k(M)$. Likewise, the spaces $H_0^k(M)=H_{0,\loc}^k(M)$ and their weighted analogues $\tau^\alpha H_0^k(M)=\{\tau^\alpha u\colon u\in H_0^k(M)\}$ can be given the structure of Hilbert spaces. Analogous constructions apply when $M$ is noncompact but one restricts to functions with support in a fixed compact subset of $M$.

The Rellich compactness theorem for compact $M$ states that the inclusion $\tau^\alpha H_0^k(M)\hra\tau^\beta H_0^l(M)$ is compact when $\alpha>\beta$ and $k>l$. This follows from the standard compactness theorem applied to an exhaustion of $M$ by smoothly bounded domains whose closures remain disjoint from $\pa M$.

%%%%%%%%%%%%%%%%%%%%%%%%%%%%%%%%%%%%%%%%%%%%%%%%%%%%%%%%%%%%%%%%%%%%%%
\section{Correction of formal solutions to true solutions}
\label{STr}

We prove a general unique continuation result (Proposition~\ref{PropTrUC}) for linear wave equations and an existence result for quasilinear wave equations (Proposition~\ref{PropTrQL}) on asymptotically de~Sitter spacetimes. We shall not optimize the smoothness requirements on source terms of waves and on the coefficients of operators.

Let $X$ be an $n$-dimensional manifold. Let $g$ denote a Lorentzian metric on $M=[0,1)_\tau\times X$ which is of the form
\begin{equation}
\label{EqTrMetric}
  g = \tau^{-2}\bar g,\qquad \bar g(\tau,x,\dd\tau,\dd x)=-\dd\tau^2+h(x,\dd x)+\bar g'(\tau,x,\dd\tau,\dd x),
\end{equation}
where $h\in\cC^1(X;S^2 T^*X)$ is a Riemannian metric on $X$ and $\bar g'\in\tau^\eta\,{}^0\cC^1(M;S^2 T^*M)$ for some $\eta>0$. In other words, $g\equiv\tau^{-2}(-\dd\tau^2+h)\bmod\tau^\eta\,{}^0\cC^1(M;S^2\,{}^0 T^*M)$. Under these assumptions, $\dd\tau$ and $-\pa_\tau$ are timelike for $\bar g$ near $\tau=0$, and we declare them to be \emph{future} timelike. Our analysis in this section will be local near points in $\pa M$. Let thus $x\in\R^n$ denote local coordinates on $X$. Let $\tau_0\in(0,\frac12]$ and consider a domain
\begin{equation}
\label{EqTrDomain}
  \Omega^{\tau_0} = \bigl\{ (\tau,x) \colon 0\leq\tau\leq\tau_0,\ |x|\leq R_0 - C\tau \bigr\},
\end{equation}
where $R_0>C\tau_0$, with $R_0>0$ so small and $C>0$ so large (and thus $\tau_0$ so small) that $\Omega^{\tau_0}$ is contained in the coordinate chart and so that $\{(\tau,x)\colon 0\leq\tau\leq\tau_0,\ |x|=R_0-C\tau\}$ and $\Omega^{\tau_0}\cap\{\tau=\tau_0\}$ are spacelike for $\bar g$. For $\delta\in(0,\tau_0]$, we further let
\begin{equation}
\label{EqTrDomSub}
  \Omega^{\tau_0}_\delta = \Omega^{\tau_0} \cap \{\tau\geq\delta\},\qquad
  \Sigma_\delta = \Omega^{\tau_0} \cap \{\tau=\delta\}.
\end{equation}
See Figure~\ref{FigTrDomain}.
\begin{figure}[!ht]
\centering
\includegraphics{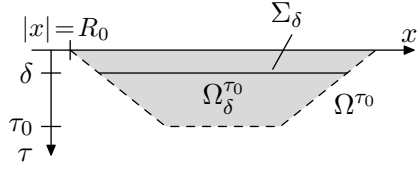}
\caption{The domain \eqref{EqTrDomain} and its subsets~\eqref{EqTrDomSub}.}
\label{FigTrDomain}
\end{figure}
We write $\CIdot(\Omega^{\tau_0})$ for the space of smooth functions on $\Omega^{\tau_0}$ (restrictions of smooth functions on $M$) which vanish to infinite order at $\tau=0$. We restrict to operators acting on real-valued functions or vectors for notational simplicity. We furthermore write
\[
  {}^0\nabla u := (\tau\pa_\tau u,\tau\pa_x u) = (\tau\pa_\tau u,\tau\pa_{x^1}u,\ldots,\tau\pa_{x^n}u),\qquad
  |{}^0\nabla u|^2 := |\tau\pa_\tau u|^2+|\tau\pa_x u|^2.
\]

%%%%%%%%%%%%%%%%%%%%%%%%%%%%%%%%%%%%%%%%%%%%%%%%%%
\subsection{Unique continuation for linear wave equations}

We prove a localized and sharpened version of \cite[Lemma~1]{ZworskiRevisitVasy}; see also \cite[Proposition~5.3]{VasyWaveOndS}. (For an elliptic analogue, see \cite{MazzeoUniqueContinuation}.)

\begin{prop}[Unique continuation]
\label{PropTrUC}
  Suppose $L=L_0+\tilde L$, with $L_0\in\Diff_0^2(M;\R^k)$ and $\tilde L\in\tau^\eta\,{}^0\cC^1\Diff_0^2(M;\R^k)$, has a scalar principal symbol given by the dual metric function $\zeta\mapsto g^{-1}(\zeta,\zeta)$. Then there exists $N<\infty$, depending only on the coefficients of $L_0$ (expressed in terms of $\tau\pa_\tau$, $\tau\pa_x$) at $\Omega^{\tau_0}\cap\{\tau=0\}$, so that every solution $u\in{}^0\cC^1(\Omega^{\tau_0};\R^k)$ which solves $L u=0$ and satisfies $|u|,|{}^0\nabla u|\leq C\tau^N$ for some $C\in\R$ must satisfy $u=0$ in $\Omega^{\tau_0}$.
\end{prop}
\begin{proof}
  By finite speed of propagation, it suffices to prove the result for some arbitrarily small but positive $\tau_0$. We restrict to the case that $k=1$ and leave the largely notational modifications needed in the general case to the reader. Consider the past timelike vector field $V_0=\tau\pa_\tau$. For $\alpha>0$ to be determined, we will use $V_\alpha=\tau^{-2\alpha}\tau\pa_\tau=\tau^{-2\alpha}V_0$ as a vector field multiplier in an energy estimate. Recall the stress-energy-momentum tensor
  \[
    T^g[u](V,W) := (V u)(W u) - \frac12 g(V,W)|\nabla^g u|^2.
  \]
  Writing $\Box_g u=-g^{j k}u_{;j k}$, we have $\dv_g T^g[u] = -(\Box_g u)\dd u$. With $\cL_{V_\alpha}g=\cL_{\tau^{-2\alpha}V_0}g=\tau^{-2\alpha}\cL_{V_0}g+2\dd(\tau^{-2\alpha})\otimes_s g(V_0,-)=\tau^{-2\alpha}\cL_{V_0}g-4\alpha\tau^{-2\alpha}\frac{\dd\tau}{\tau}\otimes_s g(V_0,-)$, we compute
  \begin{equation}
  \label{EqTrUCDiv}
  \begin{split}
    \dv_g\bigl(T^g[u](V_\alpha,-)\bigr) &= -(\Box_g u)V_\alpha u + \frac12\la\cL_{V_\alpha}g,T^g[u]\ra \\
      &= \tau^{-2\alpha} \Bigl( -(\Box_g u) V_0 u - 2\alpha T^g[u]\Bigl(\Bigl(\frac{\dd\tau}{\tau}\Bigr)^\sharp,V_0\Bigr) + \frac12\la \cL_{V_0}g,T^g[u]\ra \Bigr),
  \end{split}
  \end{equation}
  where $g((\frac{\dd\tau}{\tau})^\sharp,\cdot)=\frac{\dd\tau}{\tau}$. Since $L-\Box_g\in\Diff_0^1+\tau^\eta\,{}^0\cC^1\Diff_0^1$, we have $|((L-\Box_g)u)V_0 u|\leq C(|{}^0\nabla u|^2+|u|^2)$ where $C=C_0+\cO(\tau)+\tau^\eta C_1$, with $C_0$, resp.\ $C_1$ depending only on the coefficients of $L_0$ at $\tau=0$, resp.\ the $\tau^\eta\,{}^0\cC^1$-norm of the coefficients of $\tilde L$ in $\Omega^{\tau_0}$, with the $\cO(\tau)$ term arising from the $\tau$-dependence of the coefficients of $L_0$; thus one can take $C=2 C_0$ when $\tau_0>0$ is sufficiently small (depending on $C_1$ and the size of the $\cO(\tau)$ term). Similar estimates apply to the third term. For the second term, note that $(\frac{\dd\tau}{\tau})^\sharp\equiv-\tau\pa_\tau\bmod\tau^\eta\,{}^0\cC^1$; thus $-2\alpha T^g[u]((\frac{\dd\tau}{\tau})^\sharp,V_0)$ equals $\alpha((\tau\pa_\tau u)^2+|\tau\,{}^h\nabla u|_h^2)$ plus an error of class $(\tau^\eta\,{}^0\cC^1+\cO(\tau))|{}^0\nabla u|^2$. Choosing $\alpha$ large compared to $C$, we thus get the pointwise bound
  \[
    \dv_g\bigl(T^g[u](V_\alpha,-)\bigr) \geq \tau^{-2\alpha}\bigl(C^{-1}\alpha|{}^0\nabla u|^2 - C|u|^2\bigr).
  \]
  Since $-T^g[u](V_\alpha,(\frac{\dd\tau}{\tau})^\sharp)\sim\tau^{-2\alpha}|{}^0\nabla u|^2$ on $\Omega^{\tau_0}$, this gives (with a new constant $C$)
  \begin{equation}
  \label{EqTrUCEnergyEst}
  \begin{split}
    &\tau_0^{-2\alpha}\int_{\Sigma_{\tau_0}} |{}^0\nabla u|^2\,\frac{\dd x}{\tau_0^n} + \alpha \iint_{\Omega^{\tau_0}_\delta} \tau^{-2\alpha}|{}^0\nabla u|^2\,\frac{\dd\tau}{\tau}\frac{\dd x}{\tau^n} \\
    &\qquad \leq C\delta^{-2\alpha}\int_{\Sigma_\delta} |{}^0\nabla u|^2\,\frac{\dd x}{\delta^n} + C\iint_{\Omega^{\tau_0}_\delta}\tau^{-2\alpha}(|u|^2+|L u|^2)\,\frac{\dd\tau}{\tau}\frac{\dd x}{\tau^n}
  \end{split}
  \end{equation}
  for all $\delta\in(0,\tau_0]$. Regarding the final integral, consider a fixed value of $x$ and let $\tau_1=\min(\tau_0,C^{-1}(R_0-|x|))$; then
  \begin{align*}
    &\int_\delta^{\tau_1} \tau^{-2\alpha}|u(\tau,x)|^2\,\tau^{-n}\frac{\dd\tau}{\tau} = \int_\delta^{\tau_1} \tau^{-2\alpha}\biggl| u(\delta,x) + \int_\delta^\tau \rho^\alpha\,\rho^{-\alpha}(\rho\pa_\rho u)(\rho,x)\,\frac{\dd\rho}{\rho}\biggr|^2\,\tau^{-n}\frac{\dd\tau}{\tau} \\
    &\qquad\leq \frac12\int_\delta^{\tau_1} \tau^{-2\alpha}\biggl(|u(\delta,x)|^2 + \biggl(\int_\delta^\tau \rho^{2\alpha}\frac{\dd\rho}{\rho}\biggr)\biggl(\int_\delta^\tau \rho^{-2\alpha}|(\rho\pa_\rho u)(\rho,x)|^2\,\frac{\dd\rho}{\rho}\biggr)\biggr)\,\tau^{-n}\frac{\dd\tau}{\tau} \\
    &\qquad\leq \frac{C}{\alpha}\delta^{-2\alpha}|u(\delta,x)|^2\delta^{-n} + \frac{1}{4\alpha}\int_\delta^{\tau_1}\int_\delta^\tau \rho^{-2\alpha}|(\rho\pa_\rho u)(\rho,x)|^2\,\frac{\dd\rho}{\rho}\,\tau^{-n}\frac{\dd\tau}{\tau} \\
    &\qquad\leq \frac{C}{\alpha}\biggl( \delta^{-2\alpha}|u(\delta,x)|^2\delta^{-n} + \int_\delta^{\tau_1} \rho^{-2\alpha}|(\rho\pa_\rho u)(\rho,x)|^2\,\rho^{-n}\frac{\dd\rho}{\rho}\biggr)
  \end{align*}
  where we exchanged the $\rho$- and $\tau$-integrations in the last step. Integrating over $x$ gives
  \begin{equation}
  \label{EqTrUCPoincare}
    \iint_{\Omega^{\tau_0}_\delta} \tau^{-2\alpha}|u|^2\,\frac{\dd\tau}{\tau}\frac{\dd x}{\tau^n} \leq \frac{C}{\alpha}\biggl( \delta^{-2\alpha}\int_{\Sigma_\delta} |u|^2\,\frac{\dd x}{\delta^n} + \iint_{\Omega^{\tau_0}_\delta} \tau^{-2\alpha}|\tau\pa_\tau u|^2\,\frac{\dd\tau}{\tau}\frac{\dd x}{\tau^n} \biggr).
  \end{equation}
  Combining~\eqref{EqTrUCEnergyEst} without the first term on the left, and using $L u=0$, with \eqref{EqTrUCPoincare} gives
  \[
    \iint_{\Omega^{\tau_0}_\delta} \tau^{-2\alpha} \bigl( |u|^2 + |{}^0\nabla u|^2 \bigr)\,\frac{\dd\tau}{\tau}\frac{\dd x}{\tau^n} \leq \frac{C}{\alpha}\delta^{-2\alpha} \int_{\Sigma_\delta} \bigl( |u|^2 + |{}^0\nabla u|^2 \bigr)\,\frac{\dd x}{\delta^n}
  \]
  for all sufficiently large $\alpha$ when $\tau_0>0$ is sufficiently small. If $N$ in the statement of the proposition satisfies $-2\alpha+2 N-n>0$, then the right hand side is $o(1)$ as $\delta\searrow 0$, so letting $\delta\searrow 0$ implies that $u=0$ on $\Omega^{\tau_0}$, finishing the proof.
\end{proof}

%%%%%%%%%%%%%%%%%%%%%%%%%%%%%%%%%%%%%%%%%%%%%%%%%%
\subsection{Quasilinear waves with rapidly decaying sources}

\begin{prop}[Backwards solution of quasilinear wave equations]
\label{PropTrQL}
  Let $g$ be as in~\eqref{EqTrMetric} with $h\in\CI(X;S^2 T^*X)$ Riemannian and $\bar g'\in\tau^\eta\,{}^0\cC^\infty(M;S^2 T^*M)$. We work in local coordinates $\tau\geq 0$, $x\in\R^n$, and recall~\eqref{EqTrDomain}. Write $\ul\R^k=M\times\R^k$ for the trivial bundle. Consider a quasilinear wave operator
  \begin{equation}
  \label{EqTrQLOp}
    P(u) := \Box_{G(\tau,x;u)}u + P_1(\tau,x;u,{}^0\nabla u),
  \end{equation}
  where $u$ takes values in $\R^k$, and $G$, resp.\ $P_1$ are nonlinear bundle maps $\ul\R^k\to S^2\,{}^0 T^*M$, resp.\ $\ul\R^k\oplus\ul\R^{(n+1)k}\to\ul\R^k$ which are smooth in the fiber variables and $\CI+\tau^\eta\,{}^0\cC^\infty$ in $(\tau,x)$, and with $G(\tau,x;0)=g|_{(\tau,x)}$. Suppose that $P(0)\in\CIdot(\Omega^{\tau_0};\R^k)$. Then for sufficiently small $\tau_1\in(0,\tau_0]$, there exists a unique solution $u\in\CIdot(\Omega^{\tau_1};\R^k)$ of $P(u)=0$, and the level sets of $\tau$ as well as the boundary hypersurfaces of $\Omega^{\tau_1}$ are spacelike for $G(\tau,x;u)$.
\end{prop}
\begin{proof}
  We restrict to the case $k=1$ for notational simplicity. We write
  \[
    P_1(\tau,x;u,{}^0\nabla u) = f + L_1^\sharp u + N^\sharp(\tau,x;u,{}^0\nabla u)
  \]
  where $f=P_1(\tau,x;0,0)\in\CIdot(\Omega^{\tau_0})$, $L_1^\sharp$ is the linearization of $P_1$ around $u=0$ and thus satisfies $L_1^\sharp\in\Diff_0^1+\tau^\eta\,{}^0\CI\Diff^1_0$, and $N^\sharp(\tau,x;u,d)$ vanishes quadratically at $(u,d)=(0,0)$. Similarly, $\Box_{G(\tau,x;u)}u = -G(\tau,x;u)^{\mu\nu}\tau\pa_\mu\,\tau\pa_\nu u + L_1^\flat u + N^\flat(\tau,x;u,{}^0\nabla u)$ where $L_1^\flat,N^\flat$ are of the same type as $L_1^\sharp,N^\sharp$. With $L_1=L_1^\sharp+L_1^\flat$, $N=N^\sharp+N^\flat$, we thus wish to solve
  \begin{equation}
  \label{EqTrQLEq}
    -G(\tau,x;u)^{\bar\mu\bar\nu}\tau\pa_\mu\,\tau\pa_\nu u + L_1 u + N(\tau,x;u,{}^0\nabla u) = -f \in \CI(\Omega^{\tau_0}),
  \end{equation}
  where $G(\tau,x;u)^{\bar\mu\bar\nu}$ are the coefficients of the inverse metric of $G(\tau,x;u)$ in the frame $\tau\pa_\tau$, $\tau\pa_x$.

  We first discuss uniqueness in a stronger form than stated: suppose $u,v\in\tau^N\,{}^0\cC^3(\Omega^{\tau_1})$ both solve $P(u)=0=P(v)$, and all $\tau$-level sets and the boundary hypersurfaces of $\Omega^{\tau_1}$ are spacelike for $G(\tau,x;u)$. From~\eqref{EqTrQLEq}, we then have
  \[
    0 = P(u)-P(v) = -G(\tau,x;u)^{\bar\mu\bar\nu}\tau\pa_\mu\,\tau\pa_\nu(u-v) + \tilde L_1(u-v)
  \]
  for some operator $\tilde L_1\in\Diff_0^1+\tau^\eta\,{}^0\cC^\infty\Diff_0^1$ which depends on $u,v$ only in the lower order (in the sense of decay of coefficients at $\tau=0$) terms. Proposition~\ref{PropTrUC} thus applies to this equation when $N$ is sufficiently large \emph{independently of $u,v$}, and gives $u-v=0$, as desired.

  Turning to the question of existence, let $\delta\in(0,\tau_0]$ and denote by $u_\delta$ the solution of the initial value problem
  \begin{equation}
  \label{EqTrQLdelta}
    P(u_\delta)=0,\qquad (u_\delta,\tau\pa_\tau u_\delta)|_{\tau=\delta} = (0,0),
  \end{equation}
  which exists and is smooth on $\Omega^{\tau(\delta)}_\delta$ for some maximal $\tau(\delta)>\delta$ by standard hyperbolic theory (see e.g.\ \cite[\S16]{TaylorPDE3}); note here that the spacelike nature of the lateral hypersurfaces of $\Omega^{\tau(\delta)}_\delta$ is stable under perturbations of the metric from $g=G(\tau,x;0)$ to $G(\tau,x;u_\delta)$ since $u_\delta$ is small (pointwise) near $\tau=\delta$. The main task is the prove the existence of and uniform bounds for $u_\delta$ on $\Omega_\delta^{\tau_1}$ for some $\delta$-independent $\tau_1\in(0,\tau_0]$. As in the proof of Proposition~\ref{PropTrUC}, we apply an energy estimate using $T^{G(\tau,x;u)}[u]$ with the vector field multiplier $V_\alpha=\tau^{-2\alpha}\tau\pa_\tau$. Analogously to~\eqref{EqTrUCEnergyEst} but on $\{\delta\leq\tau\leq\rho\}$ where $\rho\leq\tau(\delta)$, and now expressing $\Box_{G(\tau,x;u_\delta)}u_\delta$ in~\eqref{EqTrUCDiv} in terms of $f$, $L_1 u_\delta$, and $N(\tau,x;u_\delta,{}^0\nabla u_\delta)$, and using the vanishing of the initial data of $u_\delta$, this gives
  \[
    \rho^{-2\alpha}\int_{\Sigma_\rho} |{}^0\nabla u_\delta|^2\,\frac{\dd x}{\rho^n} + \alpha\iint_{\Omega_\delta^\rho} \tau^{-2\alpha}|{}^0\nabla u_\delta|^2\,\frac{\dd\tau}{\tau}\frac{\dd x}{\tau^n} \leq C \iint_{\Omega_\delta^\rho} \tau^{-2\alpha}\bigl( |u_\delta|^2 + |f|^2 \bigr)\,\frac{\dd\tau}{\tau}\frac{\dd x}{\tau^n}
  \]
  for all $\alpha>1$ which are sufficiently large compared to a constant $C=C(\|u_\delta\|_{{}^0\cC^1(\Omega_\delta^\rho)})$. We use here that whenever $\|u_\delta\|_{{}^0\cC^1(\Omega_\delta^\rho)}\leq 1$, say, then $|N(\tau,x;u_\delta,{}^0\nabla u_\delta)|\leq C'(|u_\delta|^2+|{}^0\nabla u_\delta|^2)$. Upon adding an estimate of the type~\eqref{EqTrUCPoincare}, we obtain, for all sufficiently large $\alpha$, and for all $\rho\in[\delta,\tau(\delta)]$,
  \begin{align*}
    &\rho^{-2\alpha}\int_{\Sigma_\rho} \bigl( |u_\delta|^2 + |{}^0\nabla u_\delta|^2 \bigr)\,\frac{\dd x}{\rho^n} + \iint_{\Omega_\delta^\rho} \tau^{-2\alpha}\bigl(|u_\delta|^2+|{}^0\nabla u_\delta|^2\bigr)\,\frac{\dd\tau}{\tau}\frac{\dd x}{\tau^n} \\
    &\hspace{6em} \leq C\iint_{\Omega_\delta^\rho} \tau^{-2\alpha}|f|^2\,\frac{\dd\tau}{\tau}\frac{\dd x}{\tau^n}.
  \end{align*}
  (Here we also estimate $u_\delta(\rho,x)$ by writing it as the integral of $\tau\pa_\tau u_\delta(\tau,x)$ over $\tau\in[\delta,\rho]$.)

  In order to obtain higher order estimates, let $M\in\N$ and consider the equations obtained from~\eqref{EqTrQLdelta} by differentiating along all $A_{j,\beta}=(\tau\pa_\tau)^j(\tau\pa_x)^\beta$, $j+|\beta|\leq M$. The initial conditions of $A_{j,\beta}u_\delta$ at $\tau=\delta$ can be expressed entirely in terms of $f$, and their $L^2$-norms are thus uniformly bounded by any power of $\delta$. Now, applying $A_{j,\beta}$ to~\eqref{EqTrQLEq} with $u_\delta$ in place of $u$ gives
  \begin{align*}
    &-G(\tau,x;u_\delta)^{\bar\mu\bar\nu}\tau\pa_\mu\,\tau\pa_\nu(A_{j,\beta} u_\delta) \\
    &\qquad = -A_{j,\beta} f - A_{j,\beta} L_1 u_\delta - [A_{j,\beta},G(\tau,x;u_\delta)^{\mu\nu}\tau\pa_\mu\,\tau\pa_\nu]u_\delta - A_{j,\beta} N(\tau,x;u_\delta,{}^0\nabla u_\delta).
  \end{align*}
  We apply the above energy estimate to this equation, with the required lower bound on $\alpha$ depending on $M$. We proceed to estimate the $\tau^\alpha L^2$-norm of the final term (the penultimate term being handled in the same fashion) by following the method of Klainerman \cite{KlainermanUniformDecay} (see also \cite[\S6.4]{HormanderNonlinearLectures}), but mixing the usage of spacetime and spatial $L^2$-based norms here. This term is a sum of terms of the following type: a product of a partial derivative of $N$ of order $\leq M$ with at most $M$ factors of derivatives of $u_\delta$ which involve a total of at most $M$ 0-derivatives of $u_\delta$ and ${}^0\nabla u_\delta$, so less than $(M+1)$ 0-derivatives of $u_\delta$ altogether; and each term is at least quadratic in $u_\delta$ and its derivatives. In estimating this product, we use the $\tau^\alpha L^2$-norm on the at most one factor in which $u_\delta$ is differentiated at least $\frac{M+1}{2}+1$ times, and the $L^\infty$-norm on all other factors which thus involve at most $\frac{M+1}{2}$ 0-derivatives of $u_\delta$; these $L^\infty$-norms are controlled, via standard Sobolev embedding on each level set $\tau=\rho$, by the square root of $\rho^{-2 K+n}\sum_{|\beta|\leq K}\int_{\Sigma_\rho} |(\rho\pa_x)^\beta u_\delta|^2\,\frac{\dd x}{\rho^n}$ for any choice of $K>\frac{n}{2}+\frac{M+1}{2}$. Fix $M>\frac{n+1}{2}+3$ so large that an integer $K$ exists with $\frac{n}{2}+\frac{M+1}{2}<K<M+1$, and increase $\alpha$ further, if necessary, so that $\alpha>2 K-n+\frac12$ (so $\rho^{-2 K+n}\rho^{-1}\leq\rho^{-2\alpha}$ for $\rho\leq 1$). Define the quantity
  \[
    U_\delta(\tau) = \sum_{j+|\beta|\leq M+1} \tau^{-2\alpha}\int_{\Sigma_\tau} |(\tau\pa_\tau)^j(\tau\pa_x)^\beta u_\delta|^2\,\frac{\dd x}{\tau^n},
  \]
  Note that by Sobolev embedding, $U_\delta(\tau)$ controls the ${}^0\cC^1$-norm of $u_\delta$ on $\Sigma_\tau$. We therefore arrive at the estimate
  \[
    U_\delta(\tau) \leq C_N\delta^N + \int_\delta^\tau C(f,M,\alpha,U_\delta(\sigma))\cdot U_\delta(\sigma)\,\dd\sigma,
  \]
  where the (continuous) function $C$ in the integrand is independent of $\delta$. Nonnegative solutions of this ordinary differential inequality have a uniform bound on $\tau\in[\delta,\tau_1]$ for some $\delta$-independent constant $\tau_1$; indeed, taking $N=1$, we have $U_\delta(\delta)\leq C_1\delta$, and if $\bar C<\infty$ denotes a constant so that $C(f,M,\alpha,U)\leq\bar C$ when $|U|\leq 2 C_1\delta$, then by Gr\"onwall we have $U_\delta(\tau)\leq C_1\delta e^{\bar C(\tau-\delta)}\leq 2 C_1\delta$ for $\tau\leq\bar C^{-1}\log 2$ independently of $\delta$. We conclude from this the existence of $u_\delta$ on $\Omega^{\tau_1}_\delta$. As a consequence, also
  \[
    \|u_\delta\|_{\tau^{\alpha-1/2}H_0^{M+1}(\Omega_\delta^{\tau_1})}^2=\int_\delta^{\tau_1} \tau U_\delta(\tau)\frac{\dd\tau}{\tau}
  \]
  is uniformly bounded.

  Taking a sequence of $\delta$ tending to $0$, a compactness and diagonal sequence argument produces a sequence $\delta_i\searrow 0$ so that for each $i_0$ we have convergence $u_{\delta_i}\to u$ in $\tau^{\alpha-1}H_0^M(\Omega_{\delta_{i_0}}^{\tau_1})$ with $i_0$-independent bounds. Since $M>\frac{n+1}{2}+3$, this implies local $\cC^3$-convergence in $\tau>0$, and thus the limit $u$ solves the desired PDE $P(u)=0$; by Sobolev embedding we moreover have $u\in\tau^{\alpha-1}\,{}^0\cC^3(\Omega^{\tau_1})$. Choosing $\alpha$ in these arguments so that $\alpha-1>N$, the uniqueness established at the beginning of this proof applies and guarantees the equality of all subsequential limits.

  We may then apply the same arguments for larger values of $\alpha$ and $M$ and conclude that $u\in\tau^\gamma\,{}^0\cC^k(\Omega^{\tau_{\gamma,k}})$ for all $\gamma>0,k\in\N$, where $0<\tau_{\gamma,k}\leq\tau_1$. Taking $\gamma=2 k+1$ with $k$ sufficiently large, this implies $u\in\tau^k\cC^k(\Omega^{\tau_k})$ for all $k$ where $\tau_k=\tau_{2 k+1,k}$. By standard continuation results for solutions of nonlinear hyperbolic equations, the existence of $u$ in $\tau^{\alpha-1}H_0^M(\Omega^{\tau_1})\subset H^M_\loc(\Omega^{\tau_1}\setminus\tau^{-1}(0))$ together with its $H^k_\loc$-regularity on $\Omega^{\tau_k}\setminus\tau^{-1}(0)$ implies the same $H^k$-regularity also on $\Omega^{\tau_k}\setminus\Omega^{\tau_1}$. Since $k$ is arbitrary, we conclude that $u\in\bigcap_k\tau^k\cC^k(\Omega^{\tau_1})=\CIdot(\Omega^{\tau_1})$. This completes the proof.
\end{proof}

%%%%%%%%%%%%%%%%%%%%%%%%%%%%%%%%%%%%%%%%%%%%%%%%%%%%%%%%%%%%%%%%%%%%%%
\section{Proof of Theorem~\ref{ThmI}}
\label{SPf}

Given a Lorentzian metric $g_0$ in $\tau>0$, we define, following \cite{GrahamLeeConformalEinstein}, the gauge 1-form
\[
  \Ups(g) := g g_0^{-1}\delta_g\sfG_g g_0,
\]
where $(\delta_g h)_\mu=-h_{\mu\nu}{}^{;\nu}$ is the negative divergence and $\sfG_g h:=h-\frac12(\tr_g h)g$. Define the gauge-fixed Einstein vacuum operator
\[
  E(g) := 2\bigl(\Ric(g) - n g - \delta_g^*\Ups(g)\bigr).
\]
We have $E(g_0)\in\CIdot(M;S^2\,{}^0 T^*M)$ since $\Ups(g_0)=0$. Given any point $(0,x_0)\in\{0\}\in X$, choose coordinates $x\in\R^n$ near $x_0$ with $x=0$ at $x_0$.

%%%%%%%%%%%%%%%%%%%%%%%%%%%%%%%%%%%%%%%%%%%%%%%%%%
\subsection{Structure of the gauge-fixed equation}
\label{SsPfStructure}

We first claim that the equation
\begin{equation}
\label{EqPfEq}
  P(g') := E(g_0+g') = 0
\end{equation}
fits, in the chart $[0,1)_\tau\times\R^n_x$, into the setting of Proposition~\ref{PropTrQL} with $g_0$ in place of $g$ and $G(\tau,x;g')=g_0|_{(\tau,x)}+g'$, $k=\frac{(n+1)(n+2)}{2}$. (We may trivialize $S^2\,{}^0 T^*M$ by means of $\frac{\dd\tau^2}{\tau^2}$, $\frac{\dd\tau}{\tau}\otimes_s\frac{\dd x^j}{\tau}$, $\frac{\dd x^i}{\tau}\otimes_s\frac{\dd x^j}{\tau}$.) To check this, we compute, in these coordinates (with $z=(\tau,x)=(z^0,z^1,\ldots,z^n)$ and $\pa_\nu=\pa_{z^\nu}$) and for $g=g_0+g'$,
\begin{align*}
  2\Ups(g)_\mu &= 2 g_{\mu\kappa}g^{\rho\lambda}\bigl(\Gamma(g)_{\rho\lambda}^\kappa - \Gamma(g_0)_{\rho\lambda}^\kappa\bigr) \\
    &= (g_0+g')^{\rho\lambda}\bigl(2\pa_\rho(g_0+g')_{\mu\lambda}-\pa_\mu(g_0+g')_{\rho\lambda}\bigr) - 2(g_0+g')_{\mu\kappa}(g_0+g')^{\rho\lambda}\Gamma(g_0)^\kappa_{\rho\lambda}.
\end{align*}
We apply $\delta_g^*$ to this, with $(\delta_g^*W)_{\mu\nu}=\frac12(\pa_\mu W_\nu+\pa_\nu W_\mu)-\Gamma(g_0+g')^\rho_{\mu\nu}W_\rho$. The terms in $-2\delta_g^*\Ups(g)$ involving $2$ derivatives of $g'$ are
\begin{equation}
\label{EqPfDelsUps}
  (g_0+g')^{\rho\lambda} \bigr({-}\pa_\nu\pa_\rho g'_{\mu\lambda} - \pa_\mu\pa_\rho g'_{\nu\lambda} + \pa_\mu\pa_\nu g'_{\rho\lambda}\bigr).
\end{equation}
The remaining terms can be put into the term $P_1$ in~\eqref{EqTrQLOp}. To see this, write $(g_0+g')^{\bar\rho\bar\lambda}=\tau^{-2}(g_0+g')^{\rho\lambda}$, resp.\ $(g_0+g')_{\bar\rho\bar\lambda}=\tau^2(g_0+g')_{\rho\lambda}$ for the components in the frame $\tau\pa_\tau$, $\tau\pa_x$, resp.\ $\frac{\dd\tau}{\tau}$, $\frac{\dd x}{\tau}$. Then the $\bar\mu\bar\nu$ component of $\delta_g^*\Ups(g)$ includes terms such as
\[
  \tau^2\pa_\nu(g_0+g')^{\rho\lambda}\cdot\pa_\rho(g_0+g')_{\mu\lambda} = \tau^2\,\tau\pa_\nu(\tau^{-2}(g_0+g')^{\bar\rho\bar\lambda}) \tau^{-2}\,\tau\pa_\rho\bigl(\tau^2(g_0+g')_{\bar\mu\bar\lambda}\bigr);
\]
but $\tau^k \tau\pa_\nu \tau^{-k}=\tau\pa_\nu+\tau^k[\tau\pa_\nu,\tau^{-k}]$ (here needed for $k=-2,2$), with the commutator vanishing for $\nu=1,\ldots,n$ while for $\nu=0$ (so $\pa_\nu=\pa_\tau$) it is $\tau^k[\tau\pa_\tau,\tau^{-k}]=-k$ (which is smooth). All other terms can be analyzed in the same manner, as can the lower order terms (involving at most $1$ derivative of $g'$) of $\Ric(g)_{\bar\mu\bar\nu}=\tau^2\Ric(g_0+g')_{\mu\nu}$ where
\begin{alignat*}{2}
  \Ric(g_0+g')_{\mu\nu} &= \tau^2 R(g_0+g')^\rho{}_{\mu\rho\nu} \\
    &= \tau^2\bigl( \pa_\rho\Gamma(g_0+g')_{\mu\nu}^\rho - \pa_\nu\Gamma(g_0+g')_{\mu\rho}^\rho &&+ \Gamma(g_0+g')_{\mu\nu}^\kappa\Gamma(g_0+g')_{\rho\kappa}^\rho \\
    &\quad &&- \Gamma(g_0+g')_{\mu\rho}^\kappa\Gamma(g_0+g')_{\nu\kappa}^\rho\bigr).
\end{alignat*}
The terms of $2\Ric(g)_{\mu\nu}$ involving two derivatives are
\[
  (g_0+g')^{\rho\lambda}\bigl(\pa_\rho\pa_\mu g'_{\lambda\nu} - \pa_\rho\pa_\lambda g'_{\mu\nu} - \pa_\nu\pa_\mu g'_{\rho\lambda} + \pa_\nu\pa_\lambda g'_{\mu\rho}\bigr).
\]
They cancel with~\eqref{EqPfDelsUps} except for $-(g_0+g')^{\rho\lambda}\pa_\rho\pa_\lambda g'_{\mu\nu}$ which is the principal part of $\Box_{g_0+g'}(g'_{\mu\nu})=\Box_{G(\tau,x;g')}(g'_{\mu\nu})$.

%%%%%%%%%%%%%%%%%%%%%%%%%%%%%%%%%%%%%%%%%%%%%%%%%%
\subsection{Solving the Einstein vacuum equations}

Applying Proposition~\ref{PropTrQL} near every point $(0,x_0)\in\{0\}\times X\subset M$ to~\eqref{EqPfEq} (recalling that $P(0)=E(g_0)\in\CIdot(M;S^2\,{}^0 T^*M)$), we obtain a local solution $g'_{x_0}\in\CIdot(\Omega_{x_0};S^2\,{}^0 T^*M)$ on a domain $\Omega_{x_0}$ of the form~\eqref{EqTrDomain}. Since by Proposition~\ref{PropTrUC} any two local solutions $g'_{x_0}$ and $g'_{x_1}$ agree on the intersection $\Omega_{x_0}\cap\Omega_{x_1}$, these local solutions fit together to a solution $g'\in\CIdot(\cU;S^2\,{}^0 T^*M)$ of $P(g')=0$, where $\cU$ is an open neighborhood of $\{0\}\times X$ which can be taken to be the union of the interiors (in $M$) of all $\Omega_{x_0}$, $x_0\in X$; and $g=g_0+g'$ is a Lorentzian metric on $\cU$.

The second Bianchi identity for $g$, which reads $\delta_g\sfG_g\Ric(g)=0$, implies in view of $E(g)=0$ that $\Ups(g)=\Ups(g_0+g')\in\CIdot(\cU;{}^0 T^*M)$ satisfies the linear homogeneous wave equation
\[
  0 = \delta_g \sfG_g E(g) = -2\delta_g\sfG_g\delta_g^* \Ups(g).
\]
An application of Proposition~\ref{PropTrUC} gives $\Ups(g)=0$, and therefore $\Ric(g)-n g=0$.

The correction term $g'$ constructed by the above procedure is typically of the form $g'=g'(\tau,x;\dd\tau,\dd x)$. Upon pulling back $g$ by a local diffeomorphism $\Phi$ near $\{0\}\times X$ which is the identity to infinite order at $\tau=0$, one can arrange for $g'':=\Phi^*g-g_0$ to be of the form $g''=g''(\tau,x;\dd x)$; this follows from the same arguments as used in \cite[Lemma~5.2]{GrahamLeeConformalEinstein}. Relabeling $g''$ as $g'$ finishes the proof of the existence part of Theorem~\ref{ThmI}.

For the proof of uniqueness, we take $g_0$ to be equal to the true solution, i.e.\ $\Ric(g_0)-n g_0=0$. Suppose that also
\[
  g_1=g_0+h',\qquad h'=h'(\tau,x;\dd x)\in\CIdot\bigl([0,1)_\tau;\CI(X;S^2 T^*X)\bigr),
\]
satisfies $\Ric(g_1)-n g_1=0$. We shall show that $h'=0$ by working locally near any point $(0,x_0)\in\{0\}\times X\subset M$. First, there exists a smooth map $\Phi(\tau,x)=(\tau,x)+\Phi'(\tau,x)$ where $\Phi'\in\CIdot([0,1)\times\R^n;\R\times\R^n)$ so that $\Phi\colon(M,g_1)\to(M,g_0)$ is a wave map near $(0,x_0)$; similarly to the arguments in~\S\ref{SsPfStructure}, this can be seen to be a semilinear wave equation for $\Phi'$ which admits a local solution of the desired class. But then the identity map $(M,\Phi_*g_1)\to(M,g_0)$ is a wave map, which is equivalent to $\Ups(\Phi_*g_1)=0$. We conclude that $P(\Phi_*g_1-g_0)=0$, and thus $\Phi_*g_1=g_0$ by Proposition~\ref{PropTrUC}. For $x\in\R^n$ close to $x_0$, consider now the curve $\gamma\colon s\mapsto(e^s,x)$, defined for all sufficiently negative $s$. This is a geodesic for $g_1$ (with tangent vector $\tau\pa_\tau$), and thus $\Phi\circ\gamma$ is a geodesic for $g_0$ (with tangent vector $\Phi_*(\tau\pa_\tau)=\tau\pa_\tau+\cO(\tau^\infty)$). An analysis of the geodesic equation shows that the unique past-complete geodesic on $(M,g_0)$ whose tangent vector is $\tau\pa_\tau+\cO(\tau^\infty)$ is necessarily equal to $s\mapsto(a e^s,x')$, for some $a>0$, $x'\in\R^n$. Since $\Phi(0,x)=(0,x)$, we conclude that $\Phi(e^s,x)=(a(x)e^s,x)$; and since $\Phi$ equals the identity up to $\cO(\tau^\infty)$ errors, we must have $a=1$. Therefore $\Phi=\Id$, and hence $h'=g_1-g_0=g_1-\Phi_*g_1=0$, as claimed.

%%%%%%%%%%%%%%%%%%%%%%%%%%%%%%%%%%%%%%%%%%%%%%%%%%%%%%%%%%%%%%%%%%%%%%
\appendix
\section{Construction of formal solutions}
\label{SF}

We prove here the existence and uniqueness of formal solutions of~\eqref{EqI} of the form~\eqref{EqIExp} with prescribed data $h=h_0(0)$, $k=\pa_\tau^n h_0(0)$. On $M=[0,1)\times X$, with $\dim X=n\geq 3$, we introduce $e^0=\frac{\dd\tau}{\tau}$ and, following \cite[\S2.3]{HintzGluedS}, introduce bundle splittings
\begin{equation}
\label{EqFSplit}
  {}^0 T^*M = \R e^0 \oplus \tau^{-1}T^*X,\qquad
  S^2\,{}^0 T^*M = \R(e^0)^2 \oplus (2 e^0\otimes_s \tau^{-1} T^*X) \oplus \tau^{-2}S^2 T^*X.
\end{equation}
We write $(\omega_N,\omega_T)^T$, $\omega_N\in\R e^0$, $\omega_T\in T^*X$, for the 1-form $\omega_N e^0+\tau^{-1}\omega_T$, similarly for symmetric 2-tensors. We consider a metric
\[
  g=\frac{-\dd\tau^2+h(x;\dd x)}{\tau^2},
\]
which in this splitting is thus given by $(-1,0,h)^T$. We write $\dd_X$ for the exterior derivative on $X$, further $\delta_g$ and $\delta_g^*$ for the (negative) divergence and symmetric gradient, and $\sfG_g=I-\frac12 g\tr_g$. In local coordinates $x=(x^1,\ldots,x^n)\in\R^n$ on $X$, we set $e_0=\tau\pa_\tau$, $e_i=\tau\pa_{x^i}$, $e^i=\frac{\dd x^i}{\tau}$, and compute $\nabla_{e_0}e^\mu=0$, $\nabla_{e_i}e^0=h_{i k}e^k$, and $\nabla_{e_i}e^k=\delta_i^k e^0-\tau\Gamma(h)_{i j}^k e^j$. In the splittings~\eqref{EqFSplit}, this gives
\begin{alignat*}{2}
  \delta_g^*
  &=\begin{pmatrix}
     e_0 & 0 \\
     \frac12\tau\dd_X & \frac12(1+e_0) \\
     h & \tau\delta_h^*
   \end{pmatrix},&\qquad
   \delta_g
   &=\begin{pmatrix}
      e_0-n & \tau\delta_h & -\tr_h \\
      0 & e_0-n-1 & \tau\delta_h
    \end{pmatrix}, \\
  \sfG_g
  &=\begin{pmatrix}
      \frac12 & 0 & \frac12\tr_h \\
      0 & I & 0 \\
      \frac12 h & 0 & \sfG_h
    \end{pmatrix}, &\qquad
  \Box_g
  &=e_0^2 - n e_0 + \tau^2\Delta_h
  +\begin{pmatrix}
     -2 n & 4\tau\delta_h & -2\tr_h \\
     -2\tau\dd_X & -n-3 & 2\tau\delta_h \\
     -2 h & -4\tau\delta_h^* & -2
   \end{pmatrix}.
\end{alignat*}
Define the endomorphism $(\sR_g(u))_{\kappa\mu}=R^\nu{}_{\kappa\mu}{}^\rho u_{\nu\rho}+\frac12(\Ric_\kappa{}^\nu u_{\nu\mu}+\Ric_\mu{}^\nu u_{\kappa\nu})$, where $R$ and $\Ric$ are the Riemann and Ricci curvature tensors of $g$. The coefficients $R_\mu{}^\nu{}_{\kappa\lambda}=e_\mu([\nabla_{e_\kappa},\nabla_{e_\lambda}]e^\nu-\nabla_{[e_\kappa,e_\lambda]}e^\nu)$ vanish except for $R_0{}^j{}_{0 k}=-R_0{}^j{}_{k 0}=-\delta_k^j$, $R_i{}^0{}_{0 k}=-R_i{}^0{}_{k 0}=-h_{i k}$, $R_i{}^j{}_{m k}=\delta_k^j h_{i m}-\delta_m^j h_{i k}+\tau^2 R(h)_i{}^j{}_{m k}$, so
\begin{align*}
  \Ric(g)&=n g + (0,0,\tau^2\Ric(h))^T, \\
  \sR_g&=
    \begin{pmatrix}
      n & 0 & \tr_h \\
      0 & (n+1) & 0 \\
      h & 0 & (n+1)-h\tr_h
    \end{pmatrix} + \tau^2\begin{pmatrix} 0 & 0 & 0 \\ 0 & \frac12\Ric(h) & 0 \\ 0 & 0 & \sR_h \end{pmatrix}.
\end{align*}
In particular, $\Ric(g)-n g\in\tau^2\CI(M;S^2\,{}^0 T^*M)$. In order to solve away this error term by adding lower order corrections to $g$, we first compute the indicial operator of the linearization $D_g\Ric-n$: this is the bundle endomorphism obtained by acting on tensors of the form $\tau^\lambda v(x)$ where $v$ is smooth, and extracting the $\tau^\lambda$ coefficient of the result. Recalling $2 D_g\Ric=\Box_g-2\delta_g^*\delta_g\sfG_g+2\sR_g$ from \cite{DeTurckPrescribedRicci,GrahamLeeConformalEinstein}, and passing to the refined splitting
\begin{equation}
\label{EqFSplitRefined}
  S^2\,{}^0 T^*M = \R(e^0)^2 \oplus (2 e^0\otimes_s \tau^{-1}T^*X) \oplus \R\tau^{-2}h \oplus \tau^{-2}\ker\tr_h,
\end{equation}
one finds
\[
  I(D_g\Ric-n,\lambda) = \frac12\begin{pmatrix} n(\lambda-2) & 0 & -n\lambda(\lambda-2) & 0 \\ 0 & 0 & 0 & 0 \\ -(\lambda-2 n) & 0 & \lambda(\lambda-2 n) & 0 \\ 0 & 0 & 0 & \lambda(\lambda-n) \end{pmatrix}.
\]
We only consider $\lambda\geq 2$. The operator $I(D_g\Ric-n,\lambda)$ annihilates the range of
\begin{equation}
\label{EqFdelstar}
  I(\delta_g^*,\lambda) = \frac12\begin{pmatrix} 2\lambda & 0 \\ 0 & \lambda+1 \\ 2 & 0 \\ 0 & 0 \end{pmatrix}
\end{equation}
which is $2$-block-dimensional;\footnote{We regard each of the summands in~\eqref{EqFSplitRefined} to have `block-dimension' $1$; thus, the block-dimension of the range of $I(\delta_g^*,\lambda)$ is equal to the rank of the right hand side of~\eqref{EqFdelstar} regarded as a real $4\times 2$-matrix.} a subbundle $V$ complementary to the range of $I(\delta_g^*,\lambda)$ is the bundle $\tau^{-2}S^2 T^*X$ of tensors $(0,0,a,b)^T$ (in the splitting~\eqref{EqFSplitRefined}). Acting on
\[
  V:=\tau^{-2}S^2 T^*X\subset S^2\,{}^0 T^*_X M,
\]
$I(D_g\Ric-n,\lambda)$ is injective, and thus has 2-block-dimensional range, unless $\lambda=n$. On the other hand, by the second Bianchi identity (or by direct computation), the range of $I(D_g\Ric-n,\lambda)$ lies in the nullspace of
\begin{equation}
\label{EqFBianchi}
  I(\delta_g\sfG_g,\lambda) = \frac12\begin{pmatrix} \lambda-2 n & 0 & n\lambda(\lambda-2) & 0 \\ 0 & 2(\lambda-n-1) & 0 & 0 \end{pmatrix}.
\end{equation}
For $\lambda\neq n+1$, this nullspace is 2-block-dimensional. For $\lambda\neq n,n+1$, we therefore have $\ran I(D_g\Ric-n,\lambda)=\ker I(\delta_g\sfG_g,\lambda)$.

Now, $\Ric(g)-n g=\tau^2 E_2$, and by the second Bianchi identity (or by inspection, since $E_2=(0,0,\Ric(h))^T$), we have $E_2\in\ker I(\delta_g\sfG_g,2)=\ran I(D_g\Ric-n,2)|_V$. We can therefore find a unique $h_2\in\CI(X;V)\subset\CI(X;S^2\,{}^0 T^*_X M)$ so that $\Ric(g)-n g\equiv-(D_g\Ric-n)(\tau^2 h_2)\bmod\tau^3\CI$. This implies that for $g_{(2)}=g+\tau^2 h_2$ we have
\begin{equation}
\label{EqFErr2}
  \Ric(g_{(2)})-n g_{(2)}\in\tau^3\CI(M;S^2\,{}^0 T^*M).
\end{equation}
In fact, more is true: in terms of $\mu:=\tau^2$, we have $g_{(2)}=-\frac{\dd\mu^2}{\mu^2}+\mu^{-1}h_{(2)}(\mu,x;\dd x)$ where $h_{(2)}(\mu)=h+\mu h_2$, and thus $\Ric(g_{(2)})-n g_{(2)}$ is a linear combination of sections of $\R\frac{\dd\mu^2}{\mu^2}=\R(e^0)^2$, $2\frac{\dd\mu}{\mu}\otimes_s T^*X=\tau\cdot 2 e^0\otimes_s \tau^{-1}T^*X$, $S^2 T^*X=\tau^2\,\tau^{-2}S^2 T^*X$, with coefficients that are smooth in $\mu$ and thus even in $\tau$. Therefore,~\eqref{EqFErr2} is in fact even, and thus lies in $\tau^4\CI$, except for sections of $\tau^3\frac{\dd\tau}{\tau}\otimes_s \tau^{-1}T^*X$ with even coefficients in $\tau$. If $\tau^3(0,E_3,0,0)^T$ denotes the leading order part of the latter term, then the second Bianchi identity gives $(0,E_3,0,0)^T\in\ker I(\delta_g\sfG_g,3)$, which due to $n+1\neq 3$ forces $E_3=0$. Therefore,
\[
  \Ric(g_{(2)}) - n g_{(2)} \in\tau^4\CI(M;S^2\,{}^0 T^*M).
\]
Denote by $\tau^4 E_4$ the leading order term of~\eqref{EqFErr2}; by the above evenness considerations this has vanishing NT-component, by which we mean the $\frac{\dd\tau}{\tau}\otimes_s\tau^{-1}T^*X$ component. The second Bianchi identity gives $E_4\in\ker I(\delta_g\sfG_g,4)$. When $n\neq 3,4$, we can thus write $E_4=-I(D_g\Ric-n,4)h_4$ for some $h_4\in\CI(X;V)$, and so on. In this manner, we find unique $h_2,h_4,\ldots,h_{2 j}\in\CI(X;V)$ so that for $g_{(j)}=g+\tau^2 h_2+\ldots+\tau^{2 j}h_{2 j}$ we have $\Ric(g_{(j)})-n g_{(j)}\in\tau^{2 j+2}\CI(M;S^2\,{}^0 T^*M)$; for $n$ odd, we get up to $2 j=n-1$, and for $n$ even up to $2 j=n-2$.

Consider first the case that $n=2 j+1$ is odd. Then $\Ric(g_{(j)})-n g_{(j)}\equiv\tau^{n+1}E_{n+1}\bmod\tau^{n+2}\CI$, with $E_{n+1}$ lying in the 2-block-dimensional space obtained as the intersection of $\ker I(\delta_g\sfG_g,n+1)$ with the space of tensors with vanishing NT-component; since the space $\ran I(D_g\Ric-n,n+1)|_V$ is also 2-block-dimensional, one can correct $g_{(j)}$ to $g_{(j+2)}=g_{(j)}+\tau^{n+1} h_{(n+1)}$ for a unique $h_{(n+1)}\in\CI(X;V)$ so that $\Ric(g_{(j+2)})-n g_{(j+2)}\in\tau^{n+2}\CI$. Note moreover that $\tau^{-2}\ker\tr_h\subset\ker I(D_g\Ric-n,n)$; and one computes
\begin{equation}
\label{EqFDiv}
  \bigl(\tau^{-n-1}(D_g\Ric-n)(\tau^n(0,0,0,k)^T)\bigr)|_{\tau=0} = (0,-n\delta_h k,0,0)^T,\qquad k\in\CI(X;\ker\tr_h).
\end{equation}
Thus, if $k$ is transverse traceless (so $\tr_h k=0=\delta_h k$), then also $g_{(j+2)}+\tau^n(0,0,0,k)^T$ satisfies the Einstein vacuum equations up to a $\tau^{n+2}\CI$ error. From this point onward one can iteratively solve away the error term, power by power (i.e.\ not anymore in steps of two powers of $\tau$), to infinite order, which establishes~\eqref{EqIExp} for $n$ odd.

In the case that $n=2 j+2$ is even, we have $\Err_j:=\Ric(g_{(j)})-n g_{(j)}\equiv\tau^n E_n\bmod\tau^{n+1}\CI$, with $E_n$ lying in the 2-block-dimensional space $\ker I(\delta_g\sfG_g,n)$; but $I(D_g\Ric-n,n)|_V$ has 1-block-dimensional range whose projection onto the summand $\tau^{-2}\ker\tr_h$ in~\eqref{EqFSplitRefined} is trivial. Since $\pa_\lambda I(D_g\Ric-n,n)$ maps $(0,0,0,h_{1,0})^T\in\ker I(D_g\Ric-n,n)$ to $(0,0,0,n h_{1,0})^T$, we can find a unique $h_{1,0}$ so that the addition of $\tau^n(\log\tau)(0,0,0,h_{1,0})^T$ to $g_{(j)}$ corrects the error $\Err_j$ by the term $\tau^n(0,0,0,n h_{1,0})\bmod\tau^{n+1}(\log\tau)\CI+\tau^{n+1}\CI+o(\tau^{n+1})$ to a new error whose $\tau^n$-coefficient lies in $\ran I(D_g\Ric-n,n)|_V$. Thus, we can find $h_{(2 j+2)}\in\CI(X;V)$ so that for $g_{(j+1)}=g_{(j)}+\tau^n(\log\tau)(0,0,0,h_{1,0})^T+\tau^n h_{(2 j+2)}$ we have
\begin{equation}
\label{EqFErrj1}
  \Ric(g_{(j+1)})-n g_{(j+1)}\in\tau^{n+1}(\log\tau)\CI+\tau^{n+1}\CI+o(\tau^{n+1}).
\end{equation}
The NT-component of the leading order error $\tau^{n+1}((\log\tau)E_1+E_0)$ lies in the kernel of $\tau\pa_\tau-n-1$ by~\eqref{EqFBianchi}, which forces $E_1=0$. By~\eqref{EqFDiv}, adding a further term $\tau^n(0,0,0,k)^T$ to $g_{(j+1)}$ gives an additional $\tau^{n+1}(0,-n\delta_h k,0,0)^T$ contribution to~\eqref{EqFErrj1}; if one can choose $k$ so that this contribution cancels also the NT-component of $E_0$---this is the divergence equation mentioned in the introduction---, one can then proceed with the formal power series construction (with one power of $\log\tau$ until the order $\tau^{2 n}(\log\tau)^2$ due to the nonlinear nature of the Einstein vacuum equations, and so on), thus establishing~\eqref{EqIExp} for $n$ even.

%%%%%%%%%%%%%%%%%%%%%%%%%%%%%%%%%%%%%%%%%%%%%%%%%%%%%%%%%%%%%%%%%%%%%%
\bibliographystyle{alphaurl}
\bibliography{
/scratch/users/hintzp/ownCloud/research/bib/math,
/scratch/users/hintzp/ownCloud/research/bib/mathcheck,
/scratch/users/hintzp/ownCloud/research/bib/phys
}

\end{document}